\documentclass[a4paper,onecolumn]{article}

\usepackage[english]{babel}
\usepackage[utf8x]{inputenc}
\usepackage[T1]{fontenc}

\usepackage[a4paper,top=3cm,bottom=2cm,left=3cm,right=3cm,marginparwidth=1.75cm]{geometry}

\usepackage{amsmath}
\usepackage{amsfonts}
\usepackage{amssymb}
\usepackage{amsthm}

\usepackage{graphicx}
\usepackage[colorinlistoftodos]{todonotes}
\usepackage[colorlinks=true, allcolors=blue]{hyperref}

\usepackage{tikz}

\newtheorem{theorem}{Theorem}[section]
\newtheorem{corollary}{Corollary}[theorem]
\newtheorem{lemma}[theorem]{Lemma}
\newtheorem{example}[theorem]{Example}
\newtheorem{principle}{Principle}

\theoremstyle{definition}
\newtheorem{definition}{Definition}[section]

\usepackage{listings}
\usepackage[figure]{algorithm2e}

\title{Principle of Conservation of Computational Complexity}

\author{Gerald Friedland\footnote{University of California, Berkeley and Lawrence Livermore National Laboratory}, Alfredo Metere\footnote{Lawrence Livermore National Laboratory. Both authors contributed equally to this paper.}\\
fractor@eecs.berkeley.edu, metere1@llnl.gov
} 

\date{December 19th, 2017}

\begin{document}

\maketitle

\begin{abstract}
In this manuscript, we derive the principle of conservation of computational complexity. We measure computational complexity as the number of binary computations (decisions) required to solve a problem. Every problem then defines a unique solution space measurable in bits. For an exact result, decisions in the solution space can neither be predicted nor discarded, only transferred between input and algorithm. We demonstrate and explain this principle using the example of the propositional logic satisfiability problem ($SAT$). It inevitably follows that $SAT \not\in P \Rightarrow P\neq NP$. We also provide an alternative explanation for the undecidability of the halting problem based on the principle.
\end{abstract}

\maketitle

\section{Introduction}
The mathematicians Alonzo Church and Alan Turing formulated the thesis that a function on the natural numbers is computable by a human being following an algorithm, ignoring resource limitations, if and only if it is computable by a Turing machine~\cite{church1936,turing1936}. Since then, the theory of computational complexity has built upon this thesis. In this article, we will show that the neglection of resource limitations is the cause for a large amount of confusion among modern-day computer scientists. Intuitively, if one asks ``what is $3 \times 4$?'', the answer is immediately ``12''. In contrast, the response to``what is $\frac{5494380}{457865}$?'' will likely take a human much longer, despite the fact that both computations are considered constant time, commonly denoted as $\mathcal{O}(1)$. So even in an intuitive understanding of computability, there must be a notion that the same truth, encoded in different ways, will take longer to process. In our example, the second description of the number $12$ has not been fully reduced yet. In other words, the description still contains potential computations that have to be performed to reduce the description length to a minimum. 

The main focus of this article is to introduce a resource-aware interpretation of computability. We demonstrate that a computed decision is equivalent to a recorded decision (event), both of which can be measured in bits. This automatically leads to a conservation principle: A computation result that decides on an equiprobable~\cite{bayes1763} input is not predictable or avoidable~\cite{shannon1948}. The minimum number of decisions required to solve a particular problem therefore remains constant. However, decisions can be transferred from the input and therefore reduce computation. We explain why applying this conservation principle is advantageous by demonstrating it on a prominent problem~\cite{cook1971}: propositional logic satisfiability ($SAT$). We demonstrate that the cardinality of, what we call, the {\em solution space} of $SAT$ is $\mathcal{O}(n+2^n)$ bits. Therefore an exact solution to $SAT$ requires an exponential number of computation decisions when the input has only polynomial length of bits. We then set the result aside and analyze the solution space of another problem, which we call {\em Binary Codebreaker (BCB)}. One can see that $BCB$ is in $NP$. However, since the input length in bits is linear, there cannot be a polynomial algorithm for $BCB$, this is $BCB \not \in P$. Applying the principle of complexity conservation, we therefore finally find certainty that the computational complexity classes $P$ and $NP$ are not the same~\cite{pnp,clayinst}. Finally, the undecidability of the halting problem~\cite{turing1936} is a direct result of complexity conservation as well: predicting the decisions made during the execution of a program by only analyzing its syntax is not universally possible. One cannot reduce the size of the solution space to the length of the binary encoding of the program.

Formulas will be written in a notation familiar to most computer scientists~\cite{knuth2008}. We will use the words \textit{algorithm}, \textit{language} and \textit{program} interchangeably. 

\section{On the Length of a Propositional Logic Formula}
\label{sec:formulas}
A common myth among practitioners is that logic formulas are shorter than the corresponding truth tables. This is indeed commonly so but not universally, especially when formulas are encoded in binary. We start busting this myth by defining the notion of a truth table.

\begin{definition}[Truth Table ~\cite{peirce1902}]
\label{def:tt}
A truth table of a corresponding propositional logic formula is a 2-dimensional matrix with entries $\in \{0,1\}$. The matrix has one column for each variable and one final column containing the result of the evaluation of the formula when the variables are configured according to the values contained in the previous columns of the same line. On each row all combinations of the input variables and the evaluated output are listed.
\end{definition}

For convenience, we will refer to a truth table as complete truth table when it contains lines with all possible configurations and results. We will commonly refer to the results as decisions. 

\begin{definition}[Binary Decision]
\label{def:dec}
A binary decision is the result of evaluating one variable configuration of a propositional logic formula. 
\end{definition}

Naturally, it follows that a binary decision is $\in \{0,1\}$. There is no apparent reason, not to define decisions for larger alphabets as well but this is left for future work.

A truth table can be represented (encoded) in many ways -- in fact, infinitely many ways. For example, the table can be compressed with the Lempel Ziv algorithm~\cite{ziv78} or it can be learned using an Artificial Neural Network~\cite{FriedlandKrell2017}. In mathematics, a preferred way to represent a truth table is using the pre-defined functions of the Boolean Algebra~\cite{boole1854}, thus creating the very encoding we call Boolean formula~\cite{huntington1933}. In fact, algebraic reformulation of a propositional logic formula using the equivalence operator ($\Leftrightarrow$) denotes nothing else but a lossless transformation (re-coding) of the underlying truth table. If an equivalent formula is shorter, the result is a reversible reduction of the description length. In computer science and information theory, this is commonly referred to as {\em lossless compression}. 

Let $F$ and $G$ be propositional logic formulas. We denote the number of symbols needed to represent the formulas as $|F|$ and $|G|$. Furthermore, we define $\min(|F|)$ and $\min(|G|)$ as the minimum number of symbols needed to equivalently represent the formulas $F$ and $G$, respectively.
\begin{lemma}[Logic Encoding Lemma]
\label{the:lemma1}
$F \Leftrightarrow G \implies \min(|F|) = \min(|G|)$.
\end{lemma}

\begin{proof}
$F$ and $G$ can be represented as complete truth tables $T_F$ and $T_G$ respectively. $F \Leftrightarrow G \implies T_F = T_G$. Now, we can find the minimum length formula representation for $T_F$, which we call $H$. It holds that $H \Leftrightarrow F \Leftrightarrow G$. This means $|H|=\min(|F|) = \min(|G|)$.
\end{proof}

This proof does not specify a rule to generate $H$ because this is trivially not relevant. Therefore, we use {\it propositional logic formula} and {\it Boolean formula} interchangeably in this article.

We will now recite a commonly known proof, invoking the pigeon hole principle, that universal lossless compression cannot exist. 

Let $\Sigma^*$ be the set containing all the strings that can be generated from an alphabet $\Sigma = \{0,1\}$. Strings constructed from $\Sigma$ consist of {\em binary digits}, commonly referred to as {\em bits}~\cite{shannon1948}. Let the sets $A, B$ be $A \subset \Sigma^n, B \subset \Sigma^m$. $A$ contains all words of length $n$ bits and $B$ all words of length $m$ bits, this is $|A| = 2^n$ bits and $|B| = 2^m$ bits. We can now define a transformation scheme with encoding function $E: A \to B$, and decoding function $D: B \to A$. Compression is achieved iff $m < n$. The compression scheme is lossless iff the output of the decoder is equal to the input of the encoder: $D(E(A))=D(B)=A$. To achieve universal lossless compression, we would have to guarantee lossless compression for any $x \in A$. For convenience, let us define the set $U = \{ x \in A ~|~ x~ D(E(x))\neq x \}$. This is, given a compressions scheme, $U$ contains all symbols that cannot be decoded without loss. Universal lossless compression therefore implies that $U$ must be an empty set.

\begin{lemma}[No Universal Reduction of Description Length]
\label{the:lemma0}
$U \neq \emptyset, \forall~ A \subset \Sigma^* ~|~ D(E(A)) = A, |A| > |E(A)|$
\end{lemma}
\begin{proof}
Let $I \subset \Sigma^n = \{ (0), (1) \}$ hence, $|I| = 2^n, n = 1$ bit. By definition, we require a set $J \subset \Sigma^m$ such that $D(E(I))=D(J)=I$ with $m < n$. Since $n = 1$ bit, the only possible length for each element of $J$ is $m=0\,\text{bits} \implies |J| = 2^m = 1 \implies J = \{ () \}$. If $|J| = 1$ this would mean that the encoder will output 0 bits. It is self-evident that the only way for this to work is $D(J) = I$, but the only element in $J$ is 0 bits long. Therefore, $I \subset U \implies U \neq \emptyset$.
\end{proof}

For our purposes, we are now interested in the worst case minimum number of bits required to represent a propositional logic formula. We call this the {\em worst case minimum length}. Intuitively, the worst case minimum length is the minimum number of bits that need to be allocated in memory such that a program is able to allow the user to input any propositional logic formula of $n$ variables. We measure the number of bits needed to represent a set of symbols $A$ by taking the logarithm base $2$. This is, $\log_2 |A|$ is the total number of bits needed to represent all symbols in $A$ with a unique binary code. 

Let $|X|$ denote the length of a propositional logic formula $X$ and $wcm(|X|)$ the worst case minimum length of an encoding for $X$ in bits. 
\begin{lemma}
\label{the:lemma2}
$wcm(|X|)=2^n$ bits, with $n$ being the number of variables in the formula.
\end{lemma}

\begin{proof}
There are $2^{2^n}$ unique truth tables. By Lemma~\ref{the:lemma0}, we need a minimum of $\log_2(2^{2^n})=2^n$ bits to encode them. So $wcm(|X|)=2^n$ bits.
\end{proof}

\begin{example}
Let us demonstrate our thought process as follows using the Boolean Algebra on a formula of $2$ variables: there are $3$ basic operations (NOT: $\neg$, AND: $\land$, OR: $\lor$) which require at least $\log_2(3) \approx 1.585$ bits to encode, practically speaking $2$ bits. The unary negation operation describes $2^1=2$ state changes and the binary operations describe $2^2=4$ state changes each, with $2$ states being redundant ($0 \land 0 \Leftrightarrow 0 \lor 0$ and $1 \land 1 \Leftrightarrow 1 \lor 1$). This sums up to $8$ state changes. The total amount of possible state changes on $2$ variables is $2^{2^2}=16$. This means for the remaining $16-8=8$ state changes, we will need at least $2$ basic operations in the formula. Encoding $2$ operations of each $2$ bits makes $2+2=4$ bits for the operations and $1$ bit each for the variables: a total of $6$ bits. With $2$ variables we are therefore already above the worst case minimum length of $2^2=4$ bits, even if we choose to repeat this calculation without rounding to integer. 
\end{example}

As explained before, if, absurdly, one was able to compress a binary sequence of $2^n$ bits to less than $2^n$ bits universally, one could iteratively apply compression and shrink any binary sequence losslessly to $0$ bits. If there was any set of logic operators that could make a formula universally shorter than a truth table, we could encode any binary file on a computer using these operators and decompress losslessly by simply evaluating that formula. In the general case, where no external assumptions can be applied, it does not matter whether we use a classical compression algorithm for the table, a machine learning method, or the predefined functions of the Boolean Algebra: only $2^n$ bits can universally represent $2^{2^n}$ unique objects.

The human representation of a propositional logic formula is commonly smaller relative to a truth table because it uses a larger alphabet than the binary alphabet of the truth table. A typical alphabet for a Boolean formula of $n \in \mathbb{N}$ variables is $\Sigma=\{\land,\lor,\neg, x_1, ... ,x_n\}$. With a larger alphabet, $2^n$ possible states can be represented in less than $n$ symbols. In the same way, decimal number representation is shorter than binary number representation, this is $\log_{10}(2^n) \leq \log_2(2^n)$. For propositional formulas, the shortest representation can be achieved if we choose to encode the formula with an alphabet of size $2^n$. Then $\Sigma=\{F_i\}$ with $i=0,...,2^n-1$. For example, $F_{128}$ then denotes $x_1 \land x_2 \land x_3$. In this notation, we can immediately see if a formula is satisfiable, this is, if there is a configuration of the variables that evaluates to $1$: any formula that is not $F_0$ is satisfiable. However, we will demonstrate in this article that it takes exponentially many computation steps to get all formulas into this representation. 

\section{Conservation of Computational Complexity}
\label{sec:conscc}
Let us continue with the thought from the end of the previous Section. The shortest representation for propositional logic formulas can be achieved if we choose to encode the formula with $n$ variables using an alphabet of size $2^n$. Then $\Sigma=\{F_i\}$ with $i=1,...,2^n-1$. Using a different example, $F_{14}$ then denotes $x_1 \lor x_2$. We explained that, in this notation, we can immediately see if a formula is satisfiable because any formula that is not $F_0$ is satisfiable. The word {\em immediate} stands to be corrected here: it takes $O(log_b(2^n))$ comparisons to see if the formula is satisfiable or not, with $b$ being the number base we are implementing the comparisons in. For example, in base $2$, we require $n$ comparisons to be sure a formula is not satisfiable (see also Lemma~\ref{the:lemma3}). On the other end of the spectrum, it is easy to see that if the number base is chosen to be $2^n$, the search would take $O(1)$. This is a direct consequence of Lemma~\ref{the:lemma1}: we can define a formula as whatever we want -- as a truth table, as Neural Network, or even as an algorithm that can be queried by a user (see Lemma~\ref{the:theorem2}). Ultimately, however, any equivalent transformation that is universally able to encode all formulas, needs to describe the same $2^n$ bits. Intuitively: regardless of any computational assumptions, unless we are settling with an approximative result, the $2^n$ bits that a formula represents must be accessible to the computer in some way, either as input or as a decision made in computation. We will now formalize this discussion.

The standard model for studying algorithmic complexity is the Turing Machine~\cite{turing1936}. Varying definitions can be found in the literature. For the discussion in this article, it suffices to distinguish between two main classes: deterministic and non-deterministic. Non-deterministic Turing Machines can perform more than one computation step at the same time. Deterministic Turing Machines can only perform one computation step at a time. In particular, we are interested in the languages they can implement: $NP$ is the class of algorithms that can be solved on a non-deterministic Turing Machine (NTM) in polynomial time and $P$ is the class of algorithms that can be solved on a deterministic Turing Machine (DTM) in polynomial time. For almost 50 years, it has been an open question if the two complexity classes are the same. That is, if there is a way to universally transform non-deterministic polynomial algorithms into deterministic-polynomial algorithms. As a consequence, all $NP$ algorithms would be in $P$.

Without loss of generality~\cite{cook2000}, we will assume the alphabet of any Turing Machine discussed in this article to be $\Sigma=\{0,1\}$. In such a Turing Machine, all state transitions can be described by propositional logic functions. Algorithms in $NP$ are therefore allowed to have a number of computation steps bounded by $\mathcal{O}(2^{p(n)})$ while algorithms in $P$ can only have $\mathcal{O}(p(n))$ steps, where $p(n)$ is a polynomial function, this is $p(n) \leq n^k$ for a constant $k$ and $n$ the number of symbols in the input to the algorithm. We note that, for an input of zero length it is trivially impossible to establish a model of complexity based on this definition.

We are now ready to address these problems by introducing the concept of a solution space.

\begin{definition}[Solution Space]
\label{def:ss}
The solution space of a corresponding problem is the smallest multiset~\cite{knuth1998} of symbols $\in \Sigma$ that a Turing Machine must consider for an exact solution of the problem.
\end{definition}

Considering a symbol here means that it can either be read from the input or computed.

\begin{definition}[Independent Decision]
\label{def:inddec}
Let $b$ be a binary decision (see Definition~\ref{def:dec}). We call $b$ {\em independent} iff $b$ does not depend on the outcome of any other decision. 
\end{definition}

It follows that a decision can only be independent iff none of the variables in the configuration of the propositional logic formula underlying the decision are dependent on another decision. For example, $b$ would be independent if each variable depended on individual coin flips. The concept that a random coin flip defines the information content of $1$ bit has been formalized by Shannon~\cite{shannon1948}. He defined the Entropy of a discrete random variable $X$ with possible values $\{x_0, ..., x_n\}$ and probability mass function $P(X)$ as: $\mathrm {H} (X)=\mathrm {E} [-\log_2(\mathrm {P} (X))]$. The result is measured in bits. Shannon's definition is more general than we will need in this article as we will leave dependent decisions to future work. However, we demonstrate the following consistency. 

Let $b$ be a binary decision. Consistent with~\cite{shannon1948}, we call $H(b)$ the information content of $b$, measured in the unit bit. 

\begin{lemma}
\label{the:theorem2}[Equivalence of Computation and Encoding]
$b$ is independent $\Leftrightarrow H(b)$=1 bit 
\end{lemma}
\begin{proof}
Let $\Sigma=\{0,1\}$. Let $V=\{v_1, ..., v_n\}$ be the variables for a propositional logic formula $F$. Let $c \in \Sigma^{n}$ be a configuration of $V$ and $C$ be the set of all possible $2^n$ configurations. Let $F(c)=b$. 

By definition, if $b$ is independent then all of the variables in the configuration of the propositional logic formula underlying the decision are independent. Using probabilities~\cite{bayes1763}, this is $P(c_i)=P(c_j),~\forall~ 1 \leq (i,j) \leq n$, and $\sum_{m=0}^{n} P(c_m)=1$. In other words, all configurations are equiprobable. With $|C|=2^n$ it follows $H(\mathrm {P} (C))=-\log_2(2^{-n})=\log_2(2^n)=n$ bits. Now with $b$ representing $\frac{1}{n}$th of the result columns of a truth table, it follows that the information content is $\frac{n}{n}=1$ bit which we denote consistently as $H(b)=1$.

In the other direction, $H(b)$=1 bit needs to imply that $\mathrm {P}(F(b)=0)=\mathrm {P}(F(b)=1)=2^{-1}$ and thus $H({0.5,0.5})=1$. This can be easily verified using the $F_x$ notation, defined in the beginning of this section, where $x$ is the decimal representation of the result column of the truth table for $F$. Without losing generality, we will now only focus on the first line of the result column of the truth table for $F$. Since $\frac{x+1}{2}$ numbers are odd and $\frac{x+1}{2}$ numbers are even or $0$, one can verify that $\mathrm {P}(F(b)=0)=\mathrm {P}(F(b)=1)=2^{-1}$ over all truth tables represented by $F_x$. This is consistent with Definition~\ref{def:inddec}.   
\end{proof}

As explained before, assuming the alphabet of a Turing Machine $\Sigma=\{0,1\}$ implies that all state transitions can be described by propositional logic decisions. In other words, $1$ bit measures an independent decision, no matter if it is made before computation and passed as input or during computation. Independent decisions cannot be predicted or avoided without loss of accuracy. Therefore, we can refer to them as {\em irreducible}. Consequently, we can now use the bit to measure worst-case computation steps.

\begin{principle}[Conservation of Computational Complexity]
\label{pri:conservation}
Decisions in the solution space defined by a problem can neither be predicted nor discarded, only transferred between input and algorithm.
\end{principle}

\section{Satisfiability}
\label{sec:satisfiability}
We are now ready to take a closer look at satisfiability. The satisfiability problem of propositional logic ($SAT$), is the following: given a propositional logic formula $F$, is $F$ satisfiable? $F$ is satisfiable iff there is a non-empty set of configurations $C$ of the binary variables with alphabet $ \Sigma = \{0,1\}$ such that $F$ evaluates to $1$. More formally, $\exists ~c \in C = \{ b \in \Sigma ~|~ F(b) = 1\}$. This is, the language is defined as $SAT := \{code(F) \in \Sigma^* |F~\text{is a satisfiable formula of propositional logic}\}$, where $code(F)$ is the binary representation of $F$.

\begin{lemma}
\label{the:lemma3}
The size of the solution space (see Definition~\ref{def:ss}) of $SAT$ is $\mathcal{O}(n+2^n)$ bits, with $n$ being the number of variables in the formula. 
\end{lemma}

\begin{proof}
It follows from Lemma~\ref{the:lemma2} that a propositional logic formula $F$ has to be represented in at least $2^n$ bits. Since all $2^n$ decisions in the solution column are independent, the minimum number of bits it can be universally represented in is $\log_2(2^n)=n$ bits. By Lemma~\ref{the:theorem2}, we therefore need at least $n$ independent binary decisions to determine the unsatisfiability of a formula that is already fully represented in $2^n$ bits (e.g. truth table). The solution space is therefore $\mathcal{O}(n+2^n)$ bits.
\end{proof}

Before we analyze SAT further, it is helpful to remember that most formulas can be encoded in a polynomial number of bits. For the proof that follows, we are interested in patterns of unsatisfiable formulas $F_0$ that can be encoded in polynomial-length bits. Finding an example of such a polynomial pattern is straightforward. We can use a prefix bit $0$ to encode ``operator follows" and a prefix bit ``1" to encode ``number $x$ with $n$ digits follows" where $n$ is the number of independent variables in the formula and $x$ is the number of the variable. Without loss of generalization we assume that $n$ is known based on a separate transmission of that information (which takes $\log_2(n)$ bits). With three boolean operators, we therefore need $1+\log_2(3) \approx 3$ bits per operator and approximately $1+\log_2(n)$ bits per variable. Now we encode the following pattern of unsatisfiable formulas: $x_1 \land ... \land \neg x_{n-k} \land ... \land x_n $ with $0 \leq k < n$. It is easy to see that this formula pattern is bounded in length by $\mathcal{O}(n \log_2 n)$ bits and is therefore polynomial. This particular pattern alone can encode $n$ different unsatisfiable formulas. 

We are now ready to demonstrate the worst case computational complexity of $SAT$.
\begin{theorem}
\label{the:theorem11}
$SAT \not\in P$
\end{theorem}

\begin{proof}
For a proof by contradiction, let us assume $Magic(F)$, an algorithm that solves $SAT$ in polynomial time on a deterministic Turing Machine. From Lemma~\ref{the:lemma3} we know that the solution space of SAT is $n + 2^n$ bits.
A polynomial encoding of $F$ and a polynomial number of computation decisions would result in an overall polynomial count of decisions. Since $SAT$ does not allow any a-priori assumptions about the structure of the input, a universal lossless reduction (see Lemma~\ref{the:lemma0}) of the solution space would be required to implement $Magic(F)$ such that it can cope with such input universally. Per Lemma~\ref{the:lemma0}, a universal lossless compression scheme does not exist. The number of decisions in $Magic$ can therefore not be bounded by $\mathcal{O}(n^k)$. It follows $SAT \not\in P$. 
\end{proof}

Consider $Magic(F)$ to be a syntax analyzer. This is, an algorithm that uses deduction rules. In general, there is an infinite number of formulas describing one truth table. This is easily seen from the fact that every formula can, for example, be ``mirrored" infinitely with $\land$. Even though the worst case minimum length of a formula is $2^n$ bits, the maximum length is infinite. Hence, in general, we would have to match against an infinite set of finite-length  patterns that could be used to represent the same truth table: this would take infinite decisions. In other words, a syntax analysis is undecidable. We will present an additional demonstration for that in Section~\ref{sec:halting}. If we consider $Magic(F)$ a semantic analyzer, then only a reproduction of the complete truth table can lead to an exact result. This cannot be universally done in polynomial time as the truth table has an exponential number of entries. It immediately follows that the current solution of $SAT$ is actually the best case. Given a formula $F$, one needs to guess all $2^n$ variable configurations and then evaluate in linear time. 

Our result is consistent with the No-Free-Lunch theorem~\cite{nofreelunch}. It is well known that optimization needs context and cannot be universal. This is an equivalent formulation of the fact that $Magic(F)$ needs to rely on a universal lossless compression scheme (that cannot exist).

For a bigger picture, consider all unique files of length $n$ bits. There are $2^n$ such files. We now encode all files using a binary-alphabet non-deterministic Turing Machine such that each non-deterministic path encodes one file (e.g., by guessing). This defines a non-deterministic Turing Machine of size $\mathcal{O}(2^n)$ with path length $n$. That is, the path length is polynomial. For example, a file of size $n$ bits could be verified against a path in this machine in linear time. Now, from Lemma~\ref{the:lemma0} it is clear that we cannot reduce this non-deterministic Turing Machine to a polynomial-size deterministic-Turing Machine and be able to reproduce the content of all files. It immediately follows that the two machine types implement different solution spaces that are not universally reducible.

\section{Applying the Conservation Principle Directly}
\label{sec:conservation}
With the concept of the solution space, the question if $P=NP$ can also be solved more directly. Let us ignore the result described in Section~\ref{sec:satisfiability} for now. 

We now define the following problem.

\begin{definition}[Binary Codebreaker]
\label{def:bcb}
Given a set of electric switches $S$ of elements from $\{0,1\}$ with cardinality $|S|=n$. The switches are configured secretly into a code lock such that only one of the $2^n$ possible sequences $(s_1,...,s_n)$ serves as a code to unlock a door. The manufacturing company built and sold exactly all $2^n$ locks of size $n$, each with unique code. 
\end{definition}

The lock has a non-deterministic programmatic interface (reading multiple switch configurations at the same time). Also the code verification mechanism inside the lock is a non-deterministic Turing Machine. The lock and the codebreaker machine can therefore be treated as one machine. 

The obvious question we want to answer is: what is the computational complexity of opening a lock with an unknown code? The language that breaks the code to open the door is therefore $BCB := \{code \in \{0,1\}^{n} |$code opens the lock$\}$. 

\begin{lemma}
\label{the:codebreaker}
$BCB \in NP$ 
\end{lemma}

\begin{proof}
It is easy to see that $BCB$ can be verified in polynomial time as even a deterministic Turing Machine can open the door in time linear to $n$ using a binary comparison of the secret integer $i, 0 \leq i \leq 2{^n}-1$. 

A ``guess and check'' non-deterministic Turing Machine can therefore run through all $2^n$ combinations in polynomial time of the number of switches. This is, the computational steps are bounded by $\mathcal{O}(2^n)$. 
\end{proof}

\begin{lemma}
\label{the:codebreakerss}
The size of the solution space $S$ of $BCB$ is $|S|=\mathcal{O}(n+2^n)$ bits, with $n$ being the number of switches. 
\end{lemma}

\begin{proof}
There are no assumptions to be made about the configurations of switches. This is, all switch configurations are independent and so is any decision over them. In probabilistic terms, every configuration of switches $c \in (s_0,...s_n)$ has the same probability $P(c)=\frac{1}{2^n}$ to be successful. This is, the information content of a single binary decision is $H(b)=\frac{1}{2^n}$ with $b=F(c)$ for an unknown $F$. The solution space $S$ of $BCB$ is therefore the sum of all combinations ($2^n$ bits) and the time it takes to verify a solution. That is $\mathcal{O}(n)$ bits. So $|S|=\mathcal{O}(n+2^n)$ bits.  
\end{proof}

\begin{corollary}
\label{cor:bcbnotinp}
$BCB \not \in P$.
\end{corollary}
\begin{proof}
Since the input is of polynomial size and the solution space is exponential, by the principle of conservation of computational complexity, $BCB$ cannot run on a deterministic Turing Machine in polynomial time. This is $BCB \not \in P$.
\end{proof}

$BCB \in NP$ but not $\in P$. As additional note, $BCB$ seems not polynomially reducible to $SAT$ as the formula implied by $BCB$ is by definition not only unknown but also satisfiable. This implies that $BCB$ is most likely not NP-complete but this is irrelevant.

\begin{corollary}
\label{the:theorem3}
P $\not =$ NP
\end{corollary}
\begin{proof}
To show that $P \not = NP$, it suffices to show that $\exists L\in NP$ with $L \not\in P$ with $L$ being a language~\cite{cook2000}. Therefore it follows by corollary from Corollary~\ref{cor:bcbnotinp} or Theorem~\ref{the:theorem11} that $P \not = NP$.
\end{proof}

\section{On the Halting Problem}
\label{sec:halting}
We now introduce an alternative explanation for the undecidability of the halting problem~\cite{turing1936} based on the principle of conservation of computational complexity. 

A solution to a decision problem requires a certain number of irreducible decisions. Based on Lemma~\ref{the:lemma0} and Lemma~\ref{the:theorem2} it is universally impossible for a Turing Machine to skip decisions, thus reducing the solution space. If analyzing the syntax of a program of another Turing Machine requires less decisions than running the program, it would imply, again, a violation of the conservation of the solution space. We can observe the following. 

Let the language that a Turing Machine accepts be $L=\{w ~ | ~ TM ~\text{accepts}~ w\}$ with $w \in \Sigma^*$. For simplicity, we call the tape {\em memory} and allow random access by direct addressing. Furthermore, we will use a Turing-Machine-equivalent computation model, namely the $WHILE$ program~\cite{schoening} (p. 106). This model is based on Kleene's Normal Form Theorem~\cite{kleene1943}, which states that any Turing-complete program can be expressed by only one $WHILE$ loop with a Boolean condition.  The condition cannot be omitted because otherwise the program is not Turing complete ($\mu$-recursive), only $LOOP$-calculatable (primitive recursive). This has been shown by~\cite{schoening} (p.121) based on~\cite{ackermann1928}.

This means, any program can be expressed as 
\begin{verbatim}
WHILE F { P }
\end{verbatim} 
where $F$ is a function returning $0$ or $1$ and $P$ is an arbitrary sub program. If we choose the alphabet of the Turing Machine to be binary, the halting condition $F$ must be a propositional logic formula. $F$ checks a set of memory cells, modified by $P$ for the acceptance of $w$. 

More specifically, we can express the program as 
\begin{verbatim}
WHILE NOT F(cell[k],...,[m]) { P }
\end{verbatim} 
where $cell_k,...,cell_m$ are memory cells (Boolean variables) and $P$ is only responsible for modifying their values. Without losing generality, we inverted the $WHILE$ loop. $P$ ultimately only leads to the configuration of variables $cell_k,...,cell_m$. 

Deciding if $F$ is satisfiable, that is determining if there exist a configuration for $cell_k,...,cell_m$ such that $F$ is true based on $F$'s syntax (this is, without running the program), equates to solving if the program halts on some input $w$: $SAT(``F(x_k,...,x_m)")$ alone would be predicting if the $WHILE$ loop can ever halt, ignoring the values of the variables. This is undecidable and known as the Existential Halting Problem (EHP). EHP is defined as follows: given a Turing Machine $M$, is there some input $w$ on which $M$ halts, formally: $EHP:=\{< M,w > ~ | \exists~ $w$~\text{such that}~M~\text{stops on}~ w\}$. 

This reduction of $EHP$ to $SAT$ is a another way of showing that $SAT$ is only solvable by simulating all states of the independent variables. In the case of $F$ encased in a $WHILE$ program, the independence needs to be resolved fully. This is, if some variables in $F$ (memory cells) depend on other variables, they must be evaluated until each variable in $F$ is independently valid. For example, until the only dependency left is a decision by the user. For many reasons, this can take infinite time. This means, the cardinality (size) of the solution space might be infinite.

The halting problem is therefore undecidable because predictions of irreducible decisions are impossible.

\section{Conclusions}
To the best of our knowledge~\cite{pnp,clayinst}, the proofs contained in our article have not been proposed before. 
We presented the principle of conservation of computational complexity. This principle is derived by the simple fact that one bit represents two equiprobable states based on one independent binary decision. Under this viewpoint, it becomes obvious that skipping irreducible decisions is analogous to applying a lossy compression scheme that cannot guarantee an exact solution. Binarization of the input and computing space might seem strange at first. In the end, it is only a modern form of Goedelization~\cite{goedel}. We would like to note that the field of Human Computer Interaction has long sensed the existence of a law of conservation of complexity~\cite{saffer2010} that is in full agreement with our findings. Similarly, it is known in the machine learning and computer vision communities that the runtime of a signal processing algorithm for the same length input depends on the noisiness of the signal~\cite{friedlandmultimedia}. Last but not least, the consequences of the principle of the conservation of computational complexity inevitably leads to the determination that $P \neq NP$. 

In our minds, understanding complexity is an effort of reducing it. The original works of~\cite{cook1971} and~\cite{karp1972} demonstrate ingenious examples of complexity reduction. However, the apparent disadvantage of current complexity theory is that it neglects that an algorithm by itself is meaningless. An algorithm is only part of a more complex system. It is somewhat surprising that even a fantastic book like~\cite{schoening} does not even mention once the notion of a {\em bit}. The bit connects computation to information theory and physics~\cite{landauer1961,dzugutov1998}. After all, computers are part of the physical universe and bound to its principles, such as the conservation of energy. As mentioned previously, investigating how Shannon's Source {\it Coding} Theorem~\cite{shannon1948} can be used as a universal measure for computational complexity therefore seems like a viable path for future research~\cite{Vulpiani2010,FriedlandMetere2017}.

\section*{Acknowledgements}
This work was performed under the auspices of the U.S. Department of Energy by Lawrence Livermore National Laboratory under Contract DE-AC52-07NA27344. It was also partially supported by a Lawrence Livermore Laboratory Directed Research \& Development grants (17-ERD-096, 17-SI-004, and 18-ERD-021). IM number LLNL-JRNL-743757. Any findings and conclusions are those of the authors, and do not necessarily reflect the views of the funders. We want to thank our families for their support enduring weekend and late-night shifts writing this article. We also want to thank Dr. Mario Krell, Dr. Tomas Oppelstrup, Dr. Markus Schordan, Dr. Jason Lenderman, Dr. Adam Janin and Dr. Jeffrey Hittinger for encouraging remarks on this article. We want to thank Prof. Satish Rao and Prof. Richard Karp for taking the time to discuss with us. We thank our former PhD advisors, Prof. Ra\'ul Rojas and Prof. Mikhail Dzugutov as they continue to be incredible mentors. Prof. Jerome Feldman deserves thanks and credit for instrumental fundamental advise on this and other articles. We would like to point out that this article would not have been possible without the existence of Wikipedia.org. 

\bibliographystyle{alpha}
\bibliography{complexity}
\end{document}